\documentclass[aps,pra,reprint,showpacs,superscriptaddress,%
longbibliography]{revtex4-1}
\usepackage{amsfonts,amsmath,amssymb,amsthm}

\theoremstyle{definition}
\newtheorem{theorem}{Theorem}
\newtheorem{proposition}[theorem]{Proposition}
\newtheorem{corollary}[theorem]{Corollary}
\newtheorem{example}[theorem]{Example}
\newtheorem{remark}[theorem]{Remark}

\DeclareMathOperator{\bR}{\mathbb R}
\DeclareMathOperator{\bN}{\mathbb N}

\DeclareMathOperator{\Maths}{s}
\DeclareMathOperator{\MathAAE}{AAE}
\DeclareMathOperator{\sI}{\mathcal I}
\DeclareMathOperator{\sJ}{\mathcal J}
\DeclareMathOperator{\sK}{\mathcal K}

\newcommand{\sing}[1]{\Maths^\downarrow(#1)}
\newcommand{\AAE}[1]{\MathAAE^\downarrow(#1)}
\newcommand{\eigset}[2]{\lambda^\downarrow_{#2}(#1)}
\newcommand{\dg}{d_g}
\newcommand{\dpg}{d^\triangledown_g}

\begin{document}  
\title{Induced Metric And Matrix Inequalities On Unitary Matrices}

\author{H.\ F. Chau}
\email{Corresponding author, hfchau@hku.hk}
\affiliation{Department of Physics and Center of Theoretical and Computational
 Physics, University of Hong Kong, Pokfulam Road, Hong Kong}
\author{Chi-Kwong Li}
\email{ckli@math.wm.edu}
\affiliation{Department of Mathematics, College of William \& Mary,
 Williamsburg, VA 23187-8795, USA}
\altaffiliation{(in the spring of 2012) Department of Mathematics, University
 of Hong Kong, Pokfulam Road, Hong Kong}
\author{Yiu-Tung Poon}
\email{ytpoon@iastate.edu}
\affiliation{Department of Mathematics, Iowa State University, Ames, IA 50011,
 USA}
\author{Nung-Sing Sze}
\email{raymond.sze@inet.polyu.edu.hk}
\affiliation{Department of Applied Mathematics, The Hong Kong Polytechnic
 University, Hung Hom, Hong Kong}

\date{\today}

\begin{abstract}
 Recently, Chau [Quant.\ Inform.\ \& Comp.\ \textbf{11}, 721 (2011)] showed
 that one can define certain metrics and pseudo-metrics on $U(n)$, the group of
 all $n\times n$ unitary matrices, based on the arguments of the eigenvalues of
 the unitary matrices.  More importantly, these metrics and pseudo-metrics have
 quantum information theoretical meanings.  So it is instructive to study this
 kind of metrics and pseudo-metrics on $U(n)$.  Here we show that any symmetric
 norm on $\bR^n$ induces a metric on $U(n)$.  Furthermore, using the same
 technique, we prove an inequality concerning the eigenvalues of a product of
 two unitary matrices which generalizes a few inequalities obtained earlier by
 Chau [arXiv:1006.3614v1].
\end{abstract}


\maketitle 

\section{Introduction} \label{Sec:Intro}
 In quantum information science, it is instructive to measure the cost needed
 to evolve a quantum system~\cite{Lloyd} as well as to quantify the difference
 between two quantum evolutions on a system~\cite{Chefles}.
 To some extent, the solutions of both problems are closely related to certain
 pseudo-metric functions on unitary operators.  To see this, suppose we are
 given a certain quantifiable cost required to implement a unitary operation
 acting on an $n$-dimensional Hilbert space.  We may represent this cost by a
 non-negative function $f\colon U(n) \to \bR$, where $U(n)$ is the group of all
 $n\times n$ unitary matrices.  The larger the value of $f(X)$, the higher the
 cost of implementing the unitary operation $X$.  Besides, $f(X) = 0$ if it is
 costless to perform $X$.  We may further require this cost function $f$ to
 satisfy the following constraints.

\par\medskip
 \textbf{Constraints for the cost function $f$:}
 \begin{enumerate}
  \item $f(e^{i r} X) = f(X)$ for all $r\in \bR$ and $X\in U(n)$.  In addition,
   $f(I) = 0$.  The underlying reason is that changing the global phase of $X$
   has no effect on the quantum system.  Besides, the identity operation does
   not change any quantum state and hence should be costless.
  \item $f(X^{-1}) = f(X)$ for all $X\in U(n)$.  This is because $X^{-1}$ can
   be implemented by running the quantum circuit for $X$ backward in time with
   the same cost.
  \item $f(Y^{-1} X Y) = f(X)$ for all $X,Y\in U(n)$.  The rationale is that
   the cost to evolve a quantum system should be eigenbasis independent.
   Although this assumption is questionable for bipartite systems, we will
   stick to it in this paper for the evolution cost for monopartite system is
   already a worthy topics to investigate.
  \item $f(X Y) \le f(X) + f(Y)$ for all $X,Y\in U(n)$.  The reason behind is
   that a possible way to implement $X Y$ is to first apply $Y$ then follow by
   $X$.  If we further demand that the cost is additive (in the sense that the
   cost of applying $Y$ and then $X$ is equal to the cost of applying $Y$ plus
   the cost of applying $X$), which is not an unreasonable demand after all,
   then the inequality follows.
 \end{enumerate}
 A cost function $f$ induces a function $d\colon U(n) \times U(n) \to \bR$ by
 the equation $d(X,Y) = f(X Y^{-1})$ for all $X,Y\in U(n)$.  Surely, $d(X,Y)$
 can be regarded as the cost needed to transform $Y$ to $X$.  In this respect,
 the induced function $d$ provides a partial answer to the problem of
 quantifying the difference between two quantum evolutions on a system.  The
 larger the value of $d(X,Y)$, the more different the quantum operations $X$
 and $Y$ is.  More importantly, since $f$ obeys the above four constraints,
 $d(\cdot,\cdot)$ must be pseudo-metric on $U(n)$ because it satisfies $d(X,Y)
 \ge 0$, $d(X,Y) = d(Y,X)$ and $d(X,Z) \le d(X,Y) + d(Y,Z)$ for all $X,Y,Z\in
 U(n)$.  Nevertheless, $d(\cdot,\cdot)$ is not a metric for $d(X,Y) = 0$ does
 not imply $X = Y$.  We remark that the induced $d$ also obeys $d(Z X,Z Y) =
 d(X,Y)$ for all $X,Y,Z\in U(n)$.

 Conversely, suppose there is a pseudo-metric $d$ on $U(n)$ quantifying the
 difference between two unitary operations acting on a $n$-dimensional quantum
 system.  Surely, it should satisfy $d(X,X) = 0$, $d(e^{i r}X,Y) = d(X,Y)$ and
 $d(X,Y) = d(Z X,Z Y)$ for all $r\in \bR$, $X,Y,Z\in U(n)$.  The reason is that
 the difficulty in distinguishing between two unitary operations is unchanged
 by varying the global phase in one of the operations and by applying a common
 quantum operation to them.  (Again, this reason is valid as we restrict our
 study to monopartite systems.)  More importantly, $d$ induces the function
 $f(X) = d(X,I)$ on $U(n)$ which obeys the four constraints on $f$.  (The
 second constraint follows from $f(X^{-1}) = d(X^{-1},I) = d(X\, X^{-1},X) =
 d(I,X) = d(X,I) = f(X)$.  And the other three constraints can be proven in a
 similar way.)  To summarize, we have argued that the cost function $f$
 describing the resources required to evolve a (monopartite) quantum system is
 equivalent to quantifying the difference between two quantum evolutions on a
 (monopartite) system through the induced pseudo-metric function $d$.  And we
 remark on passing that our discussions so far are valid for
 infinite-dimensional quantum systems as well.

 Recently, Chau~\cite{Chau1,Chau2} introduced a family of cost functions on
 $U(n)$ based on a tight quantum speed limit lower bound on the evolution time
 of a quantum system he discovered earlier~\cite{Chau3}.  In quantum
 information science, these cost functions can be interpreted as the least
 amount of resources (measured in terms of the product of the evolution time
 and the average absolute deviation from the median of the energy) needed to
 perform a unitary operation $X\in U(n)$~\cite{Chau2}.  With the above quantum
 information science meaning in mind, it is not surprising that each cost
 function in this family depends only on the eigenvalues of its input argument
 $X$.  Actually, it can be written as a certain weighted sum of the absolute
 value of the argument of the eigenvalues of $X$~\cite{Chau1,Chau2}.

 By eigenvalue perturbation method, Chau~\cite{Chau1,Chau2} proved that for
 each cost function in the family, the corresponding induced function $d$ is
 indeed a pseudo-metrics on $U(n)$ (and hence the cost function really
 satisfies the four constraints listed earlier).  In fact, he proved something
 more.  In addition to this induced family of pseudo-metrics, he also
 discovered a family of closely related metrics on $U(n)$.  The only difference
 between them is that the family of metrics is an ``un-optimized'' version of
 the family of metrics in the sense that it does not take into account the fact
 that altering the global phase of a unitary operation does not affect the cost
 at all~\cite{Chau1,Chau2}.  More precisely, the underlying cost functions for
 the family of metrics obey the four constraints listed above except that the
 first one is replaced by $f(X) = 0$ if and only if $X = I$.  Note that given
 $X,Y\in U(n)$, the family of metrics can also be expressed as certain weighted
 sums of the absolute value of the argument of the eigenvalues of the matrix $X
 Y^{-1}$~\cite{Chau1,Chau2}.

 Interestingly, the family of metrics on $U(n)$ discovered by Chau provides
 another partial answer to the problem of quantifying the difference between
 two quantum evolutions on a system.  Specifically, Chau~\cite{Chau1,Chau2}
 showed that the metric functions he discovered can be used to give a
 quantitative measure on the degree of non-commutativity between two unitary
 matrices $X$ and $Y$ in terms of certain resources needed to transform $X Y$
 to $Y X$.

 The above background information shows that a number of quantum information
 science questions are related to the cost function $f$ (or equivalently, the
 pseudo-metric or its ``un-optimize'' metric version $d$).  Besides, the third
 constraint for $f$, namely, $f(Y^{-1} X Y) = f(X)$ for all $X,Y\in U(n)$,
 implies that the cost function $f$ depends on the eigenvalues of its input
 argument only.  Equivalently, it means that the corresponding metric and
 pseudo-metric $d(X,Y)$'s on $U(n)$ are functions of the eigenvalues of $X
 Y^{-1}$ only.

 In this paper, we adopt the following strategy to investigate the problem of
 metrics, pseudo-metrics and their relation with quantum information science.
 We begin by finding metrics and pseudo-metrics $d(X,Y)$'s on $U(n)$ that are
 functions of the eigenvalues of $X Y^{-1}$ only by means of
 Proposition~\ref{Prop:Main}.  More precisely, we prove that a symmetric norm
 of $\bR^n$ induces a metric and a pseudo-metric on $U(n)$ of the required
 type.  We then show in Example~\ref{Eg:l_p_norm} that some of the new metrics
 and pseudo-metrics discovered in this way indeed have quantum information
 science meanings.  Interestingly, Proposition~\ref{Prop:Main} has merit on its
 own for we can adapt its proof to show an inequality concerning the
 eigenvalues of a product of two unitary matrices.  This inequality is a
 generalization of several inequalities first proven in Ref.~\cite{Chau1} using
 eigenvalue perturbation technique.  Finally, we briefly discuss the
 connection of our findings and the Horn's problem on eigenvalue inequalities
 for the sum of Hermitian matrices.

\section{Metric And Pseudo-Metric Induced By A Symmetric Norm}
\label{Sec:Metric_Proof}

 To show that a symmetric norm on $\bR^n$ induces a metric and a pseudo-metric
 on $U(n)$, we make use of the following result by Thompson~\cite{T}:
\begin{theorem}[Thompson]
 If $A$ and $B$ are Hermitian matrices, there exist unitary matrices $X$ and
 $Y$ (depending on $A$ and $B$) such that
 \begin{equation}
  \exp \left( i A \right) \exp \left( i B \right) = \exp \left( i X A X^{-1} +
  i Y B Y^{-1} \right) . \label{E:exponential_formula}
 \end{equation}
 \label{Thrm:Thompson}
\end{theorem}

Note that Thompson proved his result by assuming the validity
of the Horn's conjecture concerning the relation of the
eigenvalues of the Hermitian matrices $A$, $B$, and $C = A+B$.
The Horn's conjecture was confirmed based on the works of
Klyachko~\cite{Klyachko} and Knutson and Tao~\cite{KT};
see Ref.~\cite{Fu} for an excellent survey of the results.  Later,
Agnihotri and Woodward~\cite{AW} improved the result of Thompson
and gave a necessary and sufficient condition for the eigenvalues
of (special) unitary matrices $X$, $Y$ and $Z = X Y$ using
quantum Schubert calculus.  The proof is technical and the statement
of the result involve a large set of inequalities on the arguments of
the eigenvalues of the unitary matrices  $X, Y$ and $Z=XY$
by putting them in suitable interval $[r, r+2\pi)$.
So, it is not easy to
use.  In fact, it suffices (and is actually more practical)
to use Theorem~\ref{Thrm:Thompson} to derive our results.
We will further discuss the connection between our results with the Horn's
problem in Section~\ref{Sec:Horn_Problem}.
We first present our results in the following.

Recall that a symmetric norm $g\colon \bR^n\rightarrow [0,\infty)$ is a norm
function such that $g({\bf v}) = g({\bf v}P)$ for any ${\bf v} \in \bR^{1\times
n}$, and any permutation matrix or diagonal orthogonal matrix $P$.

\begin{proposition}
 Let $g: \bR^n\rightarrow [0,\infty)$ be a symmetric norm.  We may define a
 metric on $U(n)$ as follows:
 \begin{equation}
  \dg(X,Y) = g(|a_1|, \dots, |a_n|) ,
  \label{E:d_defined}
 \end{equation}
 where $XY^{-1}$ has eigenvalues $e^{ia_j}$'s with 
 $\pi \ge a_1 \ge \cdots \ge a_n > -\pi$.  Furthermore, we may define a
 pseudo-metric on $U(n)$ by
 \begin{equation}
  \dpg(X,Y) = \inf_{r \in \bR} g(|a_1(r)|, \dots, |a_n(r)|) ,
  \label{E:d_pseudo_defined}
 \end{equation}
 where $e^{ir}XY^{-1}$ has eigenvalues $e^{ia_j(r)}$'s with 
 $\pi \ge a_1(r) \ge \cdots \ge a_n(r) > -\pi$.

 Note that the infimum above is actually a minimum as we can search 
 the infimum in any compact interval of the form $[r_0,r_0+2\pi]$.
 \label{Prop:Main}
\end{proposition}
\medskip
\begin{proof}
 Surely $\dg(X,Y),$ $\dpg(X,Y) \ge 0$ for all $X,Y\in U(n)$.  Besides,
 $\dg(X,X) = g(0,0,\dots ,0) = 0$.  And
 if $X\ne Y$, at least one eigenvalue of $XY^{-1}$ must be different from
 $1$.  Since $g$ is a norm, we conclude that $\dg(X,Y) > 0$.

 Suppose $XY^{-1}$ has eigenvalues $e^{ia_j}$'s with $\pi \ge a_1 \ge \cdots
 \ge a_n > -\pi$.  Clearly, the eigenvalues of $YX^{-1}$ are $e^{-ia_j}$'s.
 As $g$ is a symmetric norm,  $g(|a_1|, \dots, |a_n|) = g(|-a_n|, \dots,
 |-a_1|)$.  Hence, $\dg(X,Y) = \dg(Y,X)$.  By applying the same argument to
 $e^{ir} XY^{-1}$, we get $\dg(e^{ir}X,Y) = \dg(Y,e^{ir}X)$ for all $r\in
 \bR$.  From Eqs.~\eqref{E:d_defined} and~\eqref{E:d_pseudo_defined}, we know
 that $\dpg(X,Y) = \inf_{r\in \bR} \dg(e^{ir}X,Y) = \inf_{r\in \bR}
 \dg(X,e^{-ir}Y)$.  Hence, $\dpg(X,Y) = \dpg(Y,X)$.

 Finally, we verify the triangle inequalities for $\dg(\cdot,\cdot)$ and
 $\dpg(\cdot,\cdot)$.
 Let $X,Y,Z \in U(n)$. Suppose 
 $\dg(X,Y) = g(|a_1|,\dots, |a_n|)$ and 
 $\dg(Y,Z) = g(|b_1|,\dots, |b_n|)$
 where $e^{ia_1}, \dots, e^{ia_n}$ are the eigenvalues of $XY^{-1}$, and 
 $e^{ib_1}, \dots, e^{ib_n}$ are the eigenvalues of $YZ^{-1}$. Suppose
 $XZ^{-1}$ has eigenvalues $e^{ic_j}$'s with
 $\pi \ge c_1 \ge \cdots \ge c_n > -\pi$.
 Then by Theorem~\ref{Thrm:Thompson},
 there exist Hermitian matrices $A, B, C = A+B$
 with eigenvalues $a_1 \ge \cdots \ge a_n$, 
 $b_1 \ge \cdots \ge b_n$ and $\tilde c_1 \ge \cdots \ge \tilde c_n$
 such that if we replace $\tilde c_j$ by $\tilde c_j - 2\pi$ if 
 $\tilde c_j > \pi$ and replace $\tilde c_j$ by
 $\tilde c_j + 2\pi$ if $\tilde c_j \le -\pi$,
 then the resulting $n$ entries will be the same as $c_1, \dots, c_n$
 if they are arranged in descending order.
 Consequently, if 
 $\|{\bf v}\|_k$ is the sum of the $k$ largest entries of
 ${\bf v}\in \bR^{1\times n}$ for $k = 1, \dots, n$, then  
 \begin{align}
  \|(|c_1|, \dots, |c_n|)\|_k & \le \| (|\tilde c_1|, \dots,
  |\tilde c_n|)\|_k \nonumber \\
  & \le \| (|a_1|, \dots, |a_n|)\|_k +  \| (|b_1|, \dots, |b_n|)\|_k \nonumber
   \\
  & = \| (|a_1|+|b_1|,\dots ,|a_n|+|b_n|) \|_k .
  \label{E:tmp1}
 \end{align}
 Note that to arrive at the second inequality above, we have used the fact
 that
 \begin{equation}
  \|M+N\|_k \le \|M\|_k + \|N\|_k
  \label{E:Ky_Fan_inequality}
 \end{equation}
 for any $n\times n$ complex-valued matrices $M, N$ and for $k = 1, \dots, n$.
 Here $\|M\|_k$ is the Ky Fan $k$-norm, which is defined as the sum of 
 the $k$ largest singular values of $M$~\cite{F}.

 Since $g({\bf u}) \le g({\bf v})$ for any ${\bf u}, {\bf v} \in
 \bR^{1\times n}$ if and only if
 $\|{\bf u}\|_k \le \|{\bf v}\|_k$ for $k = 1, \dots, n$~\cite{FH,LT},
 it follows that 
 \begin{align}
  \dg(X,Z) & \le g(|c_1|, \dots, |c_n|) \nonumber \\
  & \le g(|a_1|+|b_1|,\dots, |a_n|+|b_n|) \nonumber \\
  & \le g(|a_1|, \dots, |a_n|) + g(|b_1|, \dots, |b_n|) \nonumber \\
  & = \dg(X,Y)+\dg(Y,Z) .
  \label{E:tmp2}
 \end{align}

 Since the infimum in Eq.~\eqref{E:d_pseudo_defined} is actually a minimum,
 there exist $r(X,Y), s(Y,Z)\in \bR$ such that $\dpg(X,Y) = \dg(e^{ir}X,Y)$
 and $\dpg(Y,Z) = \dg(e^{is}Y,Z) = \dg(Y,e^{-is}Z)$.  From Eq.~\eqref{E:tmp2},
 \begin{align}
  \dpg(X,Y) + \dpg(Y,Z) &= \dg(e^{ir}X,Y) + \dg(Y,e^{-is}Z) \nonumber \\
  &\ge \dg(e^{ir}X,e^{-is}Z) \nonumber \\
  &= \dg(e^{i(r+s)}X,Z) \nonumber \\
  &\ge \dpg(X,Z) .
  \label{E:tmp3}
 \end{align}
 The proof is complete.
\end{proof}
      
\medskip\noindent
\begin{example}
 For any ${\mathbf\mu} = (\mu_1, \dots, \mu_n) \in \bR^n$,
 define the ${\mathbf \mu}$-norm by
 \begin{equation}
  \|{\bf v}\|_{\mathbf\mu} =  \max\left\{\sum_{j=1}^n |\mu_j v_{i_j}| \colon
  \{i_1,\dots,i_n\} = \{1, \dots, n\}\right\} . \label{E:nu_norm_defined}
 \end{equation}
 Clearly this is a family of symmetric norms; and the induced metrics and
 pseudo-metric on $U(n)$ are the families of metrics and pseudo-metrics
 introduced by Chau in Refs.~\cite{Chau1,Chau2}.
 \label{Eg:mu_norm}
\end{example}

\begin{example}
 One may pick $g$ to be the $\ell_p$ norm defined by
 $\ell_p({\bf v}) = \left(\sum_{j=1}^n |v_j|^p \right)^{1/p}$ for any $p \in
 [1, \infty]$.  The induced metric on $U(n)$ has some interesting quantum
 information science meanings.  In fact, it will be shown in Ref.~\cite{LC}
 that this induced metric is a new family of indicator functions on the minimum
 resources needed to perform a unitary transformation.  Moreover, these
 indicator functions are closely related to a new set of quantum speed limit
 bounds on time-independent Hamiltonians~\cite{LC} generalizing the earlier
 results by Chau~\cite{Chau1,Chau2,Chau3}.
 \label{Eg:l_p_norm}
\end{example}

\begin{remark}
 In the perturbation theory context, we consider
 $\tilde{X} = X E$, where $E$ is very close to the identity.
 Suppose $X = e^{iA}$, where $A$ has eigenvalues
 $\pi-\varepsilon > a_1 \ge \cdots \ge a_n > -\pi + \varepsilon$, 
 and  $E = e^{iB}$ such that the eigenvalues of $B$ lie in 
 $[-\varepsilon, \varepsilon]$
 for an $\varepsilon > 0$. 
 Then we may conclude that $\tilde{X}$ has eigenvalues 
 $\pi > c_1 \ge \cdots \ge c_n > -\pi$ such that
 $|c_j - a_j| \le \varepsilon$.
 \label{Rem:perturbation_context}
\end{remark}

\section{Several Inequalities On Products Of Two Unitary Matrices}
\label{Sec:Matrix_Product_Inequalities}

 The proof technique used in Proposition~\ref{Prop:Main} can be used to show
 an inequality generalizing a few similar ones originally reported by Chau in
 Ref.~\cite{Chau1}.

 First, recall that given two non-increasing sequences of real numbers
 ${\bf u} = (u_1,\dots,u_n)$ and ${\bf u}' = (u_1',\dots,u_n')$, we say that
 ${\bf u}$ is weakly sub-majorized by ${\bf u}'$ if $\sum_{j=1}^k u_j \le
 \sum_{j=1}^k u_j'$ for $1\le k\le n$.  Furthermore, a real-valued function
 $h({\bf u})$ is said to be Schur-convex if $h({\bf u}) \le h({\bf u}')$
 whenever ${\bf u}$ is weakly sub-majorized by ${\bf u}'$.

\begin{proposition}
 Let
 \begin{equation}
  h(\sing{A + B},\sing{A},\sing{B}) \le 0 \label{E:Hermitian_base_eq}
 \end{equation}
 be an inequality valid for all $n$-dimensional Hermitian matrices $A$ and
 $B$, where $\sing{A}$ denotes the sequence of singular values of $A$
 arranged in descending order.  Suppose further that $h$ is a Schur-convex
 function of its first argument whenever the second and third arguments are
 kept fixed.  Then,
 \begin{equation}
  h(\AAE{X Y},\AAE{X},\AAE{Y}) \le 0 \label{E:unitary_resultant_eq}
 \end{equation}
 where $\AAE{X}$ denotes the sequence of absolute value of the principal value
 of argument of the eigenvalues of an $n\times n$ unitary matrix $X$ arranged
 in descending order.
 In other words, if the eigenvalues of the unitary matrix $X$ are $e^{i a_1},
 \dots e^{i a_n}$ with $a_j \in (-\pi,\pi]$ for all $j$ and $|a_1| \ge |a_2|
 \ge \dots \ge |a_n|$, then $\AAE{X} = (|a_1|,|a_2|,\dots,|a_n|)$.
 \label{Prop:additional}
\end{proposition}
\begin{proof}
 Let $X, Y \in U(n)$.  And write $X = \exp ( i A )$,
 $Y = \exp (i B)$ and $X Y = \exp (i C)$ where the eigenvalues of the
 Hermitian matrices $A, B, C$ are all in the range $(-\pi,\pi]$.  By
 Theorem~\ref{Thrm:Thompson}, we can find a Hermitian matrix $\tilde{C}$ and $X
 Y = \exp ( i \tilde{C} )$, where $\tilde{C} = W_1 A W_1^{-1} + W_2 B W_2^{-1}$
 for some $W_1, W_2 \in U(n)$.  Hence,
 $h(\sing{\tilde{C}},\sing{A},\sing{B}) \leq 0$.

 Note that the eigenvalues of $\tilde{C}$ need not lie on the interval
 $(-\pi,\pi]$.  Yet, we can transform $\tilde{C}$ to $C$ by replacing those
 eigenvalues $a_j$'s of $\tilde{C}$ by $a_j + 2\pi$ if $a_j \leq -\pi$ and
 replacing them by $a_j - 2\pi$ if $a_j > \pi$.  Obviously, $\sing{C}$ is
 weakly sub-majorized by $\sing{\tilde{C}}$.  Therefore,
 \begin{align}
  & h(\AAE{X Y},\AAE{X},\AAE{Y}) \nonumber \\
  = & h(\sing{C},\sing{A},\sing{B})
   \le h(\sing{\tilde{C}},\sing{A},\sing{B}) \le 0 .
  \label{E:proof_2_tmp}
 \end{align}
 So, we are done.
\end{proof}

\begin{corollary}
 Let $X, Y \in U(n)$ and that $X$, $Y$ and $X Y$ have eigenvalues
 $e^{i a_j}$'s, $e^{i b_j}$'s and $e^{i c_j}$'s, respectively with $\pi \ge
 |a_1| \ge \cdots \ge |a_n| \ge 0$, $\pi \ge |b_1| \ge \cdots \ge |b_n| \ge 0$
 and $\pi \ge |c_1| \ge \cdots \ge |c_n| \ge 0$.  Then
 \begin{equation}
  \sum_{\ell = 1}^p |c_{j_\ell+k_\ell-\ell}| \le \sum_{\ell=1}^p \left(
  |a_{j_\ell}| + |b_{k_\ell}| \right) , \label{E:Lidskii1}
 \end{equation}
 for any $1\le j_1 < \cdots < j_p \le n$ and $1\le k_1 < \cdots < k_p \le n$
 with $j_p + k_p - p \leq n$.
 \label{Cor:additional}
\end{corollary}
\begin{proof}
 Eq.~\eqref{E:Lidskii1} is the direct consequences of
 Proposition~\ref{Prop:additional} and the inequality
 \begin{equation}
  \sum_{\ell = 1}^p \eigset{A+B}{j_\ell+k_\ell-\ell} \le
  \sum_{\ell = 1}^p \left[ \eigset{A}{j_\ell} + \eigset{B}{k_\ell} \right]
  \label{E:Lidskii_Hermitian_inequality1}
 \end{equation}
 reported in Ref.~\cite{Talt}.  Here $\eigset{A}{j}$ denotes the $j$th
 eigenvalue of the Hermitian matrix $A$ arranged in descending order.
\end{proof}

\begin{remark}
 Actually, Eq.~\eqref{E:Lidskii_Hermitian_inequality1} belongs to a class of
 matrix inequalities in the form
 \begin{equation}
  \sum_{k\in\sK} \eigset{A+B}{k} \le \sum_{i\in\sI} \eigset{A}{i} +
  \sum_{j\in\sJ} \eigset{B}{j} ,
  \label{E:Lidskii_Hermitian_type_inequality}
 \end{equation}
 where $A, B$ are $n\times n$ Hermitian matrices and $\sI,\sJ,\sK$ are subsets
 of $\{1,2,\dots,n \}$ with equal cardinality.  This class of matrix
 inequalities is sometimes called the Lidskii-type inequalities.  Thus,
 Proposition~\ref{Prop:additional} implies that every Lidskii-type inequality
 for Hermitian matrix induces a corresponding inequality for unitary matrix.
 \label{Rem:induced_Lidskii-type_inequalities}
\end{remark}

\section{Relation To The Horn's Problem} \label{Sec:Horn_Problem}

 In fact, Lidskii-type inequalities are closely related to the Horn's problem
 in matrix theory.  Horn~\cite{Horn} conjectured that eigenvalues of the $n
 \times n$ Hermitian matrices $A$, $B$ and $A+B$ are completely characterized
 by inequalities in the form Eq.~\eqref{E:Lidskii_Hermitian_type_inequality}
 and the equality
\begin{equation}
 \sum_{j=1}^n \eigset{A+B}{j} = \sum_{j=1}^n \left[ \eigset{A}{j} +
 \eigset{B}{j} \right] .
 \label{E:det_constraint}
\end{equation}
 (That is to say, he believed that eigenvalues of $A$, $B$ and $A+B$ obey
 Eq.~\eqref{E:det_constraint} and certain Lidskii-type inequalities.
 Furthermore, given three decreasing sequences of real numbers
 $( a_j )_{j=1}^n$, $( b_j)_{j=1}^n$ and $( c_j )_{j=1}^n$ satisfying
 $\sum_{j=1}^n c_j = \sum_{j=1}^n \left( a_j + b_j \right)$ and the
 corresponding Lidskii-like inequalities in the form $\sum_{k\in\sK} c_k \le
 \sum_{i\in\sI} a_i + \sum_{j\in\sJ} b_j$, then there exist Hermitian matrices
 $A$, $B$ and $A+B$ whose eigenvalues equal $a_j$'s, $b_j$'s and $c_j$'s,
 respectively.)  Horn also wrote down a highly inefficient inductive algorithm
 to find the subsets $\sI$, $\sJ$ and $\sK$~\cite{Horn}.  The Horn's problem
 was proven by combined works of Klyashko~\cite{Klyachko} and Knutson and
 Tao~\cite{KT}.  Besides, the existence of a minimal set of Lidskii-type
 inequalities for the Horn's problem was also
 shown~\cite{Klyachko,KT,Belkale,KTW}.  In this regard,
 Remark~\ref{Rem:induced_Lidskii-type_inequalities} can be restated as follow:
 each of the minimal set of Lidskii-type inequalities for the Horn's problem
 induces an inequality for the eigenvalues of unitary matrices $X$, $Y$ and
 $X Y$.

 Naturally, one asks if these corresponding inequalities completely
 characterizes the eigenvalues of the product of unitary matrices.  This
 problem, which is sometimes called the multiplicative version of the Horn's
 problem, was solved by the combined works of Agnihorti and Woodward~\cite{AW}
 and Belkale~\cite{Belkale,Belkale2} by means of quantum Schubert calculus.
 Phrased in the content of our current discussion, they proved the following.
 Let $e^{2\pi i\alpha_j}$'s, $e^{2\pi i\beta_j}$'s and $e^{2\pi i\gamma_j}$'s be
 eigenvalues of the $n\times n$ special unitary matrices $X$, $Y$ and $Z$,
 respectively.  Surely, one may constrain the phases of the eigenvalues by
 $\sum_{j=1}^n \alpha_j = 0$ and $\alpha_1 \ge \alpha_2 \ge \dots \ge
 \alpha_n > \alpha_1 - 1$.  And $\beta_j$'s and $\gamma_j$'s are similarly
 constrained.  Then, the eigenvalues of $X$, $Y$ and $Z$ satisfying $X Y Z = I$
 are completely characterized in the sense of the Horn's problem by
 inequalities in the form
\begin{equation}
 \sum_{i\in \tilde{\sI}} \alpha_i + \sum_{j\in \tilde{\sJ}} \beta_j +
 \sum_{k\in \tilde{\sK}} \gamma_k \le d
 \label{E:multiplicative_Horn_constraints}
\end{equation}
 for some $d(\tilde{\sI},\tilde{\sJ},\tilde{\sK})\in \bN$ known as the
 Gromov-Witten invariant, where the subsets $\tilde{\sI}$, $\tilde{\sJ}$ and
 $\tilde{\sK}$ of $\{ 1,2,\dots , n\}$ are of the same cardinality.  Similar to
 the Horn's problem, only a highly inefficient recursive algorithm is known to
 date to find these inequalities.  Thus, it is instructive to see how to deduce
 our induced inequalities from those completely characterizing the
 multiplicative version of the Horn's problem as this problem seems to be
 non-trivial.  In fact, a major difficulty of this approach is the different
 ways to order the eigenvalues $e^{ia_j}$'s --- ours are ordered by the values
 of $|a_j|$'s while those arising from the multiplicative version of the Horn's
 problem are ordered by the values of $a_j$'s.
 Note that in applications of matrix inequalities to practical problems such as
 numerical analysis and perturbation theory, it is often the case that one can
 deduce the useful results using the basic Lidskii-type inequalities in the
 form of Eqs.~\eqref{E:Lidskii1} or~\eqref{E:Lidskii_Hermitian_inequality1},
 and rarely would one use the full generalizations in
 Eq.~\eqref{E:Lidskii_Hermitian_type_inequality}.  In fact, specializing the
 general results in (14) to deduce well known matrix inequalities may actually
 be quite involved.  For example, see Theorem~3.4 and the discussion after it
 in Ref.~\cite{LP03}.  In that paper, we obtained our main results using
 Thompson's theorem efficiently. As mentioned before, it will be instructive to
 use the general inequalities of the multiplicative version of Horn's problem
 to deduce our results, but it may not be easy and not very practical.

\begin{acknowledgments}
 We like to thank K.-Y.\ Lee for pointing out a mistake in our draft.
 H.F.C.\ is supported in part by the RGC grant HKU~700709P of the HKSAR
 Government.  Research of C.K.L.\ is supported by a USA NSF grant, a HK RGC
 grant, and the 2011 Shanxi 100 Talent Program.  He is an honorary professor of
 University of Hong Kong, Taiyuan University of Technology, and Shanghai
 University.  Research of Y.T.P.\ is supported by a USA NSF grant and a HK RGC
 grant.  Research of N.S.S.\ is supported by a HK RGC grant.
\end{acknowledgments}

\bibliography{qc52.3}
\end{document}